\pdfoutput=1
\documentclass{toc}
\usepackage{petar-common-toc}

\runningtitle{Electric routing and concurrent flow cutting}
\tocpdftitle{Electric routing and concurrent flow cutting}

\runningauthor{Jonathan Kelner and Petar Maymounkov}
\copyrightauthor{Jonathan Kelner and Petar Maymounkov}
\tocpdfauthor{Jonathan Kelner and Petar Maymounkov}

\begin{document}
\begin{frontmatter}
\title{Electric routing and concurrent flow cutting}
\author[jon]{Jonathan Kelner
\thanks{Research partially supported by NSF grant CCF-0843915.}}
\author[petar]{Petar Maymounkov \footnotemark[\value{footnote}]}

\begin{abstract}
We investigate an oblivious routing scheme, amenable to distributed computation
and resilient to graph changes, based on electrical flow.  Our main technical
contribution is a new rounding method which we use to obtain a bound on the
$\ell_1\to\ell_1$ operator norm of the inverse graph Laplacian. We show how 
this norm reflects both latency and congestion of electric routing.
\end{abstract}
\tockeywords{oblivious routing, spectral graph theory, Laplacian operator}
\tocacm{F.2.2, G.2.2}
\tocams{68Q25, 68R10}
\end{frontmatter}

\newcommand{\kron}{\chi}
\newcommand{\imbal}{\zeta}
\newcommand{\Li}{L^{\dag}}
\newcommand{\el}{{\mathcal{E}}}
\newcommand{\fopt}{\mathrm{F}_{\mathrm{opt}}}
\newcommand{\cn}{{G}}
\renewcommand{\opt}{\mathrm{opt}}
\newcommand{\pot}{\varphi}
\newcommand{\mix}{\mathrm{mix}}

\section{Introduction}
\label{sec:intro}

\paragraph{Overview}
We address a vision of the Internet where every participant exchanges messages
with their direct friends and no one else. Yet such an Internet should be able
to support reliable and efficient routing to remote locations identified by
unchanging names in the presence of an ever changing graph of connectivity. 

Modestly to this end, this paper investigates the properties of routing along
the electric flow in a graph (\textit{electric routing} for short) for intended
use in such distributed systems, whose topology changes over
time. We focus on the class of expanding graphs which, we believe, gives a good
trade-off between applicability and precision in modeling a large class of
real-world networks.  We address distributed representation and 
computation.  As a measure of performanace, we show that electric routing,
being an oblivious routing scheme, achieves minimal maximum edge congestion (as
compared to a demand-dependent optimal scheme). Furthermore, 
we show that electric
routing continues to work (on average) in the presence of large edge failures
ocurring after the routing scheme has been computed,
which attests to its applicability in changing environments. We now proceed to
a formal definition of oblivious routing and statement of our results.

\paragraph{Oblivious routing}
The object of interest is a graph $G=(V,E)$ (with $V=[n]$ and $|E|=m$)
undirected, positively edge-weighted by $w_{u,v}\ge 0$, and not necessarily
simple.  The intention is that higher $w_{u,v}$ signifies stronger connectivity
between $u$ and $v$; in particular, $w_{u,v}=0$ indicates the absence of edge
$(u,v)$.  For analysis purposes, 
we fix an arbitrary orientation ``$\to$'' on the edges $(u,v)$ of $G$, i.e.
if $(u,v)$ is an edge then exactly one of $u\to v$ or $v\to u$ holds.
Two important operators are associated to every $G$.
The \textit{discrete gradient} operator $B\in\matha{R}^{E\times V}$, sending
functions on $V$ to functions on the \textit{undirected} edge set $E$, is
defined as $\kron_{(u,v)}^* B := \kron_u -\kron_v$ if $u\to v$, and
$\kron_{(u,v)}^* B := \kron_v -\kron_u$ otherwise, where $\kron_y$ is the
Kronecker delta function with mass on $y$.  For $e\in E$, we use the
shorthand $B_e:= (\kron_e B)^*$.  The \textit{discrete divergence} operator is
defined as $B^*$.

A \textit{(single-commodity) demand} of 
amount $\alpha > 0$ between $s\in V$ and $t\in V$ is
defined as the vector $d=\alpha(\kron_s-\kron_t)\in\matha{R}^V$. A
\textit{(single-commodity) flow}
on $G$ is defined as a vector $f\in\matha{R}^E$, so that $f_{(u,v)}$ equals the
flow value from $u$ towards $v$ if $u\to v$, and the negative of this value
otherwise.  We also use the notation $f_{u\to v}:=f_{(u,v)}$ if $u\to v$, and
$f_{u\to v}:=-f_{(u,v)}$ otherwise. We say that flow $f$ \textit{routes} demand
$d$ if $B^*f=d$. This is a linear algebraic way of encoding the fact that
$f$ is an $(s,t)$-flow of amount $\alpha$.
A \textit{multi-commodity demand}, also called a \textit{demand set}, is a
matrix whose columns constitute the individual demands' vectors. It is given as
the direct product $\oplus_\tau d_\tau$ of its columns.
A \textit{multi-commodity flow} is represented as a matrix $\oplus_\tau
f_\tau$, given as a direct product of its columns, the single-commodity
flows. For clarity, we write $f_{\tau,e}$ for $(f_\tau)_e$.
The flow $\oplus_\tau f_\tau$ \textit{routes} the demand 
set $\oplus_\tau d_\tau$ if 
$B^*f_\tau = d_\tau$, for all $\tau$, or in matrix notation
$B^*(\oplus_\tau f_\tau)=\oplus_\tau d_\tau$.
The \textit{congestion} $\|\cdot\|_\cn$ of a multi-commodity 
flow measures the load
of the most-loaded edge, relative to its capacity. It is given by
\begin{align}
\|\oplus_\tau f_\tau\|_\cn 
  := \max_e \sum_\tau \big|f_{\tau,e}/w_e\big|
  =\|(\oplus_\tau f_\tau)^*W^{-1}\|_{1\to 1},
  \text{ where }\|A\|_{1\to 1}:=\sup_{x\neq 0} \frac{\|Ax\|_1}{\|x\|_1}.
  \label{eq:congnorm}
\end{align}  

An \textit{oblivious routing scheme} is a (not necessarily linear) 
function $R:\matha{R}^V\to\matha{R}^E$ which has the property that $R(d)$
routes $d$ when $d$ is a valid single-commodity demand 
(according to our definition). We extend $R$ to a function over demand sets
by defining $R(\oplus_\tau d_\tau):=\oplus_\tau R(d_\tau)$. This says that
each demand in a set is routed independently of the others by its corresponding
$R$-flow. 
We measure the ``goodness'' of an oblivious routing scheme by
the maximum traffic that it incurs on an edge (relative to its capacity)
compared to that of the optimal (demand-dependent) routing.
This is captured
by the \textit{competitive ratio $\eta_R$} of the routing scheme $R$, defined
\begin{align}
\eta_R := 
    \sup_{\oplus_\tau d_\tau}
    \sup_{\substack{\oplus_\tau f_\tau \\ 
    B^* (\oplus_\tau f_\tau) = \oplus_\tau d_\tau}}
    \frac{\|R(\oplus_\tau d_\tau)\|_\cn}{\|\oplus_\tau f_\tau\|_\cn}.
    \label{def:eta}
\end{align}
Let $\el$ denote the (yet undefined) function corresponding to electric
routing. Our main theorem states:
\begin{theorem}\label{theo:main}
For every undirected graph $G$ with unit capacity edges, maximum degree
$d_{\max}$ and \textit{vertex expansion}
$\alpha:=\min_{S\subseteq V}\frac{|E(S,S^\cmpl)|}{\min\{|S|,|S^\cmpl|\}}$, 
one has
$\eta_{\el} \le 
\Big(4\ln \frac{n}{2}\Big) \cdot
  \Big(\alpha\ln \frac{2d_{\max}}{2d_{\max}-\alpha}\Big)^{-1}
$. This is tight up to a factor of $O(\ln\ln n)$.
\end{theorem}
The competitive ratio in Theorem~\ref{theo:main} is best achievable for any
oblivious routing scheme up to a factor of $O(\ln\ln n)$ due to a lower bound
for expanders, i.e. the case $\alpha = O(1)$, given in~\cite{leighton}.
Theorem~\ref{theo:main} can be extended to other definitions of graph
expansion, weighted and unbounded-degree graphs.  We omit these extensions for
brevity. We also give an unconditional, albeit much worse, bound on $\eta_\el$:

\begin{theorem}\label{theo:cong:univ}
For every unweighted graph on $m$ edges, 
electrical routing has $\eta_\el \le O(m^{1/2})$.
Furthermore, there are families of graphs with corresponding demand sets 
for which $\eta_\el = \Omega(m^{1/2})$.
\end{theorem}

\paragraph{Electric routing}
Let $W=\vdiag(\dots,w_e,\dots)\in\matha{R}^{E\times E}$ be the edge weights
matrix.  We appeal to a known connection between graph Laplacians and electric
current~\cite{snell, spielman-lecture}. Graph edges are viewed as wires of
resistance $w_e^{-1}$ and vertices are viewed as connection points.  If
$\pot\in\matha{R}^V$ is a vector of vertex potentials then, by Ohm's law, the
\textit{electric flow} (over the edge set) 
is given by $f=WB\pot$ and the corresponding demand is
$B^*f=L\pot$ where the \textit{(un-normalized) Laplacian} $L$ is defined as
$L=B^*WB$. Central to the present work will be the
vertex potentials that induce a desired $(s,t)$-flow, given by
$\pot^{[s,t]}=\Li(\kron_s-\kron_t)$, where $\Li$ is the pseudo-inverse of $L$. 
Thus, the electric flow corresponding to the
demand pair $(s,t)$ is $WB\pot^{[s,t]}=WB\Li(\kron_s-\kron_t)$.
We define the \textit{electric routing} operator as
\begin{align}
\el(d) = WB\Li d \label{eq:el}
\end{align}
The vector $\el(\kron_s-\kron_t)\in\matha{R}^E$ encodes a unit flow from 
$s$ to $t$ supported on $G$, where the flow along an edge $(u,v)$ is given by
$\llbracket st,uv \rrbracket
:=\el(\kron_s-\kron_t)_{u\to
v}=(\pot^{[s,t]}_u-\pot^{[s,t]}_v)w_{u,v}$.\footnote{The bilinear form 
$\llbracket st,uv \rrbracket = \kron_{s,t} B\Li B^* \kron_{u,v}$
acts like a ``representation'' of $G$, hence the custom bracket notation.}
(Our convention is that current flows towards lower potential.)
When routing an indivisible message (an IP packet e.g.), we can view the
unit flow $\el(\kron_s-\kron_t)$ as a distribution over $(s,t)$-paths 
defined recursively as follows: Start at $s$.
At any vertex $u$, \textit{forward} the message
along an edge with positive flow, with probability proportional to the
edge flow value. Stop when $t$ is reached. 
This rule defines the \textit{electric walk} from $s$ to $t$.
It is immediate that the flow value over an edge $(u,v)$ equals the 
probability that the electric walk traverses that edge.

Let ``$\sim$'' denote the vertex adjacency relation of $G$. In order to make
a (divisible or indivisible) forwarding decision, a vertex $u$ must be able to
compute $\llbracket st,uv \rrbracket$ for all neighbors 
$v\sim u$ and all pairs $(s,t)\in\binom{V}{2}$. We address this next.

\paragraph{Representation}
In order to compute $\llbracket st,uv \rrbracket$ (for all $s,t \in V$ and 
all $v\sim u$) at $u$, it suffices that $u$ stores the vector 
$\pot^{[w]}:=\Li \kron_w$, for all $w\in\{w:w\sim u\}\cup\{u\}$. This is
apparent from writing
\begin{align}
\llbracket st,uv \rrbracket
  = (\kron_u-\kron_v)\Li(\kron_s-\kron_t) 
  = (\pot^{[u]}-\pot^{[v]})^*(\kron_s-\kron_t),
  \label{fwd:flow}
\end{align}
where we have (crucially) used the fact that $\Li$ is symmetric. The
vectors $\pot^{[w]}$ stored at $u$ comprise the \textit{(routing) table}
of $u$, which consists of $\deg(u)\cdot n$ real numbers.
Thus the per-vertex table sizes of our scheme grow linearly with the vertex
degree \textendash\ a property we call \textit{fair} representation. 
It seems that fair representation is key for routing in heterogenous sytems
consisting of devices with varying capabilities.

Equation~(\ref{fwd:flow}),
written as $\llbracket st,uv \rrbracket=(\kron_s-\kron_t)^*(\pot^{[u]}-\pot^{[v]})$,
shows that in order to compute $\llbracket st,uv \rrbracket$ at $u$, it suffices
to know the \textit{indices} of $s$ and $t$ (in the $\pot^{[w]}$'s). 
These indices could be represented by $O(\ln n)$-bit opaque vertex ID's 
and could be carried in the message headers. Routing schemes that support
opaque vertex addressing are called \textit{name-independent}. Name
independence allows for vertex name persistence across time (i.e. changing
graph topology) and across multiple co-existing routing schemes.

\paragraph{Computation}
We use an idealized computational model to facilitate this exposition.  The
vertices of $G$ are viewed as processors, synchronized by a global step
counter. During a time step, pairs of processors can exchange messages of
arbitrary (finite) size as long as they are connected by an edge.
We describe an algorithm for computing
approximations $\tilde{\pot}^{[v]}$ to 
all $\pot^{[v]}$ in $O(\ln n/\lambda)$ steps, where $\lambda$ is
the Fiedler eigenvalue of $G$ (the smallest non-zero eigenvalue of $L$). 
If $G$ is an expander, then $\lambda=O(1)$.
At every step the algorithm sends messages consisting of
$O(n)$ real numbers across every edge and performs $O(\deg(v)\cdot n)$ 
arithmetic operations on each processor $v$.
Using standard techniques, this algorithm can be
converted into a relatively easy-to-implement 
asynchronous one. (We omit this detail from here.)
It is assumed that no graph changes occur during the computation of 
vertex tables. 

A vector $\zeta\in\matha{R}^V$ is \textit{distributed} if $\zeta_v$ is
stored at $v$, for all $v$. A matrix $M\in\matha{R}^{V\times V}$ 
is \textit{local} (with respect to $G$)
if $M_{u,v}\neq 0$ implies $u\sim v$ or $u = v$. It is straightforward that
if $\zeta$ is distributed and $M$ is local, 
then $M\zeta$ can be computed in a single step, resulting in a new distributed
vector. Extending this technique shows that for any polynomial $q(\cdot)$,
the vector $q(M)\zeta$ can be computed in $\deg(q)$ steps.

The Power Method gives us a matrix polynomial $q(\cdot)$ of degree $O(\ln n /
\lambda)$ such that $q(L)$ is a ``good'' approximation of $\Li$.  We compute
the distributed vectors $\zeta^{[w]}:=q(L)\kron_w$, for all $w$, in parallel.
As a result, each vertex $u$ obtains
$\tilde{\pot}^{[u]}=(\zeta^{[1]}_u,\dots,\zeta^{[n]}_u)$, which approximates
$\pot^{[u]}$ according to Theorem~\ref{coro:power} and the symmetry of $L$.
In one last step, every processor $u$ sends 
$\tilde{\pot}^{[u]}$ to its neighbors. The approximation error $n^{-5}$ is
chosen to suffice (in accordance with Corollary~\ref{coro:approx}) 
as discussed next.

\begin{theorem}\label{coro:power}
Let $\lambda$ be the Fiedler (smallest non-zero) eigenvalue of $G$'s Laplacian
$L$, and let $G$ be of bounded degree $d_{\max}$. Then
$\|\zeta^{[v]}-\pot^{[v]}\|_2\le n^{-5}$, where 
$\zeta^{[v]} = (2d_{\max})^{-1} \sum_{\omega=0}^k M^\omega \kron_v$ 
and $M=I-L/2d_{\max}$, as long as $k \ge \Omega(\lambda^{-1}\cdot \ln n)$.
\end{theorem}

\paragraph{Robustness and latency}

In order to get a handle on the analysis of routing in an ever-changing network
we use a simplifying assumption: the graph does not change during the
computation phase while it can change afterwards, during the routing phase.
This assumption is informally justified because the computation phase in
expander graphs (which we consider to be the typical case) is relatively fast,
it takes $O(\ln n)$ steps. The routing phase, on the other hand, should be as
``long'' as possible before we have to recompute the scheme.  Roughly, a
routing scheme can be used until the graph changes so much from its shape when
the scheme was computed that both the probability of reaching destinations and
the congestion properties of the scheme deteriorate with respect to the
new shape of the graph. We quantify the robustness of electric routing against
edge removals in the following two theorems:

\begin{theorem}\label{theo:robust1}
Let $G$ be an unweighted graph with Fiedler eigenvalue $\lambda=\Theta(1)$
and maximum degree $d_{\max}$,
and let $f^{[s,t]}$ denote the unit electric flow between $s$ and $t$.
For any $0 < p \le 1$, let $Q_p =\{e\in E: |f^{[s,t]}_e|\ge p\}$ 
be the set of edges carrying more than $p$ flow.
Then, 
$|Q_p|\le \min\{ 2 / (\lambda p^2), 2d_{\max}\|\Li\|_{1\to 1}/p\}$.
\end{theorem}

Note that part one of this theorem, i.e. $|Q_p|\le 2/(\lambda p^2)$,
distinguishes electric routing from simple schemes like shortest-path routing.
The next theorem studies how edge removals affect demands when ``the entire
graph is in use:''

\begin{theorem}\label{theo:robust2}
Let graph $G$ be unweighted of bounded degree $d_{\max}$ and 
vertex expansion $\alpha$.
Let $f$ be a routing of the uniform multi-commodity demand set over $V$ 
(single unit of demand between every pair of vertices),
produced by an $\eta$-competitive oblivious routing scheme. %
Then,
for any $0\le x\le 1$, removing a $x$-fraction of edges from $G$
removes at most 
$x \cdot (\eta \cdot d_{\max}\cdot \ln n\cdot \alpha^{-1})$-fraction 
of flow from $f$.
\end{theorem}

The expected number of edges traversed between source and sink
reflects the latency of a routing.
We establish (Proven in Appendix~\ref{sec:latency-proof}):
\begin{theorem}\label{theo:stretch}
The latency of every electric walk on an undirected 
graph of bounded degree $d_{\max}$ 
and vertex expansion $\alpha$
is at most $O(\min\{m^{1/2},d_{\max}\alpha^{-2}\ln n \})$.
\end{theorem}

\paragraph{Analysis}

The main hurdle is Theorem~\ref{theo:main}, which we attack in two steps.
First, we show that any linear routing scheme $R$ (i.e. scheme for which
the operator $R:\matha{R}^V\to\matha{R}^E$ is linear) has a distinct worst-case
demand set, known as \textit{uniform demands}, consisting of a unit demand
between the endpoints of every edge of $G$. Combinging this with the formulaic
expression for electric flow~(\ref{eq:el}) gives us an operator-based geometric
bound for $\eta_\el$, which in the case
of a bounded degree graph is simply $\eta_\el \le O(\|\Li\|_{1\to 1})$ 
where the operator norm $\|\cdot\|_{1\to 1}$ is defined by
$\|A\|_{1\to 1} := \sup_{x\neq 0} \|Ax\|_1 / \|x\|_1$.
This is shown in Theorem~\ref{theo:eta:ub}. Second, we give a rounding-type
argument that establishes the desired bound on $\|\Li\|_{1\to 1}$. This
argument relies on a novel technique we dub \textit{concurrent flow cutting}
and is our key technical contribution. This is done in Theorem~\ref{theo:ve:ub}.
This concludes the analysis of the congestion properties of electric flow.

The computational procedure for the vertex potentials $\pot^{[v]}$'s (above)
only affords us approximate versions $\tilde{\pot}^{[v]}$ with $\ell_2$ error
guarantees.  We need to ensure that, when using these in place of the exact
ones, all properties of the exact electric flow are preserved.  For this
purpose, it is convenient to view the electric flow as a
distribution over paths (i.e. the electric walk, defined above) 
and measure the total variation distance between the walks induced
by exact and approximate vertex potentials.
This is achieved in Theorem~\ref{theo:approx} and Corollary~\ref{coro:approx}. 
It is then easy to verify that any two multi-commodity flows, whose respective
individual flows have sufficiently small variation distance, have essentially
identical congestion and robustness properties.

\paragraph{Related work}
Two bodies of prior literature concern themselves with oblivious routing.
One focuses on approximating the shortest-path 
metric~\cite{tz1, tz2, stretch-tradeoff, compact-stretch}, the
other focuses on approximating the minimal congestion universally across
all possible demand sets~\cite{cong-opt, cong-avg}. The algorithms
in these works are essentially best possible in terms of competitive
characteristics, however they are not distributed and do not address
(competitive) performance in the presence of churn.
It is not obvious how to provide efficient distributed variants for 
these routing schemes that are additionally resistant to churn. 
The primary reason for this are the algorithmic
primitives used.  Common techniques are landmark (a.k.a. beacon)
selection~\cite{tz1, tz2}, hierarchical tree-decomposition or
tree-packings~\cite{cong-opt}. These approaches place disproportiantely
larger importance on ``root'' nodes, which makes the resulting schemes 
vulnerable to individual failures. Furthermore, these algorithms
require more than (quasi-)linear time (in the centralized model), which
translates to prohibitively slow distributed times.

We are aware of one other work in the theoretical literature 
by Goyal, et al.~\cite{goyal}
that relates to efficient and churn-tolerant distributed routing.
Motivated by the proliferation of routing schemes for trees, they
show that expanders are well-approximated by the union of $O(1)$
spanning trees. However, they do not provide a routing scheme, since
routing over unions of trees is not understood.

Concurrently with this paper, Lawler, et al.~\cite{lp-obliv} study 
just the congestion of electric flow in isolation from other considerations like
computation, representation or tolerance to churn. 
Their main result is a variant of
our graph expansion-based bound on $\|\Li\|_{1\to 1}$, given by
Theorem~\ref{theo:ve:ub}. Our approaches, however, are different.
We use a geometric approach, compared to a less direct
probabilistic one. Our proof exposes structural information about 
the electric flow, which makes
the fault-tolerance of electric routing against edge removal an easy
consequence. This is not the case for the proofs found in~\cite{lp-obliv}.

\paragraph{Organization}

Section~\ref{sec:prelim} covers definitions and preliminaries.
Section~\ref{sec:cong} relates the congestion of electric flow to
$\|\Li\|_{1\to 1}$. Section~\ref{sec:l1} obtains a bound on $\|\Li\|_{1\to 1}$
by introducing the concurrent flow cutting method.  
Section~\ref{sec:walk} relates the electric walk to the electric flow and
proves (i) stability against perturbation of vertex potentials, (ii) latency
bounds and (iii) robustness theorems. 
Section~\ref{sec:concl} contains remarks and open problems.

\section{Preliminaries}
\label{sec:prelim}


The Spectral Theory of graphs is comprehensively covered in~\cite{chung}.
The object of interest is a graph $G=(V,E)$ (with $V=[n]$ and $|E|=m$)
undirected, positively edge-weighted by 
$w_{u,v}\ge 0$, and not necessarily simple. Whenever we use unweighted
graphs, we have $w_{u,v}=1$ if $u\sim v$ and $w_{u,v}=0$ otherwise.
The Laplacian is positive semi-definite, $L\succeq 0$, and thus can
be diagonalized as $L=U\Lambda U^*$, where $U\in \matha{R}^{n\times n}$ is
unitary and $\Lambda=\vdiag(\lambda_1,\dots,\lambda_n)$. By convention, we
write $0\le\lambda_1\le\lambda_2\le\cdots\le\lambda_n$. For every $G$,
$\lambda_n\le 2D$, where $D$ is the maximum degree. 
When $G$ is connected, $\vker(L) = \vi$, and so
$L\Li=\Li L = \pi_{\bot\vi}$ where $\pi_W$ denotes projection onto $W$
and $\Li$ denotes the pseudo-inverse of $L$.
On occasion we use 
$\lambda_{\min}:=\lambda_2$ and $\lambda_{\max} := \lambda_n$.
The \textit{vertex expansion} of an unweighted $G$ is defined as 
$\alpha:=\min_{S\subseteq V}\frac{|E(S,S^\cmpl)|}{\min\{|S|,|S^\cmpl|\}}$.

\section{The geometry of congestion}
\label{sec:cong}

Recall that given a multi-commodity demand, electric routing
assigns to each demand the corresponding electric flow in $G$,
which we express~(\ref{eq:el}) in operator form
$\el(\oplus_\tau d_\tau) := WB\Li(\oplus_\tau d_\tau)$.
Electric routing is oblivious, since 
$\el(\oplus_\tau d_\tau)=\oplus_\tau\el(d_\tau)$ ensures that
individual demands are routed independently from each other.
The first key step in our analysis, Theorem~\ref{theo:eta:ub}, 
entails bounding $\eta_\el$ 
by the $\|\cdot\|_{1\to 1}$ matrix norm of a certain natural graph
operator on $G$. This step hinges on the observation that
all linear routing schemes have an easy-to-express worst-case demand set:

\begin{theorem} \label{theo:eta:ub}
For every undirected, weighted graph $G$, let $\Pi=W^{1/2}B\Li B^*W^{1/2}$,
then
\begin{align}
\eta_\el \le \|W^{1/2}\Pi W^{-1/2}\|_{1\rightarrow 1}.
    \label{ineq:eta:ub}
\end{align}
\end{theorem}


\begin{proof}[Proof of Theorem~\ref{theo:eta:ub}]
It is sufficient to consider demand sets that can be routed
in $G$ with unit congestion, since both electric and optimal routing
scale linearly with scaling the entire demand set.
Let $\oplus_\tau d_\tau$ be any demand set,
which can be (optimally) routed in $G$ with unit congestion
via the multi-commodity flow $\oplus_\tau f_\tau$.
Thus, $d_\tau=\sum_e f_{\tau,e}B_e$, for all $\tau$.

The proof involves two steps:
\begin{align*}
\big\|\el(\oplus_\tau d_\tau)\big\|_\cn
    \overset{\text{(i)}}{\le} 
      \big\|\el(\oplus_e w_e B_e)\big\|_\cn
    \overset{\text{(ii)}}{=} 
      \big\|W^{1/2}\Pi W^{-1/2}\big\|_{1\rightarrow 1} 
\end{align*}

Step (i) shows that congestion incurred when routing $\oplus_\tau d_\tau$
is no more than that incurred when routing $G$'s edges, viewed as demands, 
through $G$:
\begin{align*}
\big\|\el(\oplus_{\tau}d_\tau)\big\|_\cn
  &= \big\|\oplus_\tau \el(d_\tau)\big\|_\cn \tag{i} \\
  &= \big\|\oplus_\tau\el\big(\sum_e f_{\tau,e}B_e\big)\big\|_\cn
  &\quad\text{use $d_\tau=\sum_e f_{\tau,e}B_e$} \\
  &= \big\|\oplus_\tau\sum_e \el(f_{\tau,e}B_e)\big\|_\cn
  &\quad\text{use $\el\big(\sum_j d_j\big)=\sum_j\el(d_j)$} \\
  &\le \big\|\oplus_{\tau,e}\el(f_{\tau,e}B_e)\big\|_\cn
  &\quad\text{use 
    $\big\|\sum_j f_j\big\|_\cn \le \big\|\oplus_j f_j\big\|_\cn$} \\
  &= \big\|\oplus_e\el\big(\sum_\tau|f_{\tau,e}|B_e\big)\big\|_\cn
  &\quad\text{use $\big\|\oplus_j \alpha_j f\big\|_\cn 
      =\big\|\sum_j |\alpha_j| f\big\|_\cn$} \\
  &\le \big\|\oplus_e \el(w_e B_e)\big\|_\cn
  &\quad\text{use $\sum_\tau|f_{\tau,e}|\le w_e$} \\
  &= \big\|\el(\oplus_e w_e B_e)\big\|_\cn
\end{align*}

\begin{align*}
\big\|\el(\oplus_e w_e B_e)\big\|_\cn
  &\overset{(\ref{eq:congnorm})}{=} 
    \big\|\el(\oplus_e w_e B_e)^*W^{-1}\big\|_{1\to 1} 
  \tag{ii}\\
  &\overset{(\ref{eq:el})}{=} 
    \big\|WB\Li B^*W W^{-1}\big\|_{1\to 1} 
    = \|W^{1/2}\Pi W^{-1/2}\|_{1\to 1}. \notag
    &&\qedhere
\end{align*}

\end{proof}

\begin{remark} 
Note that the proof of step (i) uses only the linearity of $\el$ and so
it holds for any linear routing scheme $R$, i.e. one has 
$\|R(\oplus_\tau d_\tau)\|_G\le \|R(\oplus_e w_e B_e)\|_G$.
\end{remark}

Using Theorem~\ref{theo:eta:ub}, the unconditional upper bound in
Theorem~\ref{theo:cong:univ} is simply a consequence of basic norm
inequalities. (See Appendix~\ref{sec:cong:u} for a proof.)
Theorem~\ref{theo:main} provides a much stronger bound on $\eta_\el$
when the underlying graph has high vertex expansion.
The lower bound in Theorem~\ref{theo:main} is due to Hajiaghayi, et
al.~\cite{leighton}. They show that every oblivious routing scheme is bound to
incur congestion of at least $\Omega(\ln n/ \ln\ln n)$ on a certain 
family of expander graphs. 
The upper bound in Theorem~\ref{theo:main} 
follows from Theorem~\ref{theo:eta:ub}, Theorem~\ref{theo:ve:ub} and using that 
$\|\Pi\|_{1\to 1}=O(\|\Li\|_{1\to 1})$ for unweighted bounded-degree graphs.
Thus in the next section we derive a bound on $\|\Li\|_{1\to 1}$ in terms of 
vertex expansion.

\section{$L_1$ operator inequalities}
\label{sec:l1}

The main results here are an upper and lower bound on $\|\Li\|_{1\to 1}$, which
match for bounded-degree expander graphs. In this section, we present vertex
expansion versions of these bounds that assume bounded-degree. 


\begin{theorem}\label{theo:ve:ub}
Let graph $G=(V,E)$ be unweigthed, of bounded degree $d_{\max}$, and
vertex expansion 
\begin{align}\label{ve:iso}
\alpha = \min_{S\subseteq V}\frac{|E(S,S^\cmpl)|}{\min\{|S|,|S^\cmpl|\}},
\quad\text{then}\quad
\|\Li\|_{1\to 1} 
  \le \Big(4\ln \frac{n}{2}\Big) \cdot
  \Big(\alpha\ln \frac{2d_{\max}}{2d_{\max}-\alpha}\Big)^{-1}.
\end{align}
\end{theorem}

The proof of this theorem (given in the next Section) boils down to a
structural decomposition of unit $(s,t)$-electric flows in a graph (not
necessarily an expander).  We believe that this decomposition is of independent
interest.  In the case of bounded-degree expanders, one can informally say that
the electric walk corresponding to the electric flow between $s$ and $t$ takes
every path with probability exponentially inversely proportional to its length.
We complement Theorem~\ref{theo:ve:ub} with a lower bound on $\|\Li\|_{1\to 1}$
proven in Appendix~\ref{sec:norm:lower}:

\begin{theorem}\label{theo:m:lb}
Let graph $G=(V,E)$ be unweighted, of bounded degree $d_{\max}$,
with metric diameter $D$.
Then, $\|\Li\|_{1\to 1} \ge 2D/d_{\max}$
and, in particular, 
\mbox{$\|\Li\|_{1\to 1} 
\ge \big(2\ln n\big)\cdot\big(d_{\max}\ln d_{\max}\big)^{-1}$}
for all bounded-degree, unweighted graphs with vertex expansion $\alpha=O(1)$.
\end{theorem}


\subsection{Proof of upper bound on $\|\Li\|_{1\to 1}$ for expanders}

\begin{proof}[Proof of Theorem~\ref{theo:ve:ub}] 

\textit{Reformulation:}
We start by transforming the problem in a more manageable form,
\begin{align}\label{ineq:ve:form}
\|\Li\|_{1\to 1} := \sup_{y\neq 0} \frac{\|\Li y\|_1 }{ \|y\|_1} 
= \max_{w} \|\Li\kron_w\|_1 \overset{(*)}{\le} 
\frac{n-1}{n}\max_{s\neq t} \|\Li(\kron_s-\kron_t)\|_1,
\end{align}
where the latter inequality comes from
\begin{align*}
\|\Li \kron_s \|_1 
  = \|\Li\pi_{\bot\vi} \kron_s\|_1
  = \|n^{-1}\sum_{t\neq s}\Li (\kron_s-\kron_t)\|_1
  \le \frac{n-1}{n} \max_{t} \|\Li (\kron_s-\kron_t)\|_1.
\end{align*}

Pick any vertices $s\neq t$ and set $\psi=\Li(\kron_s-\kron_t)$.  In light of
(\ref{ineq:ve:form}) our goal is to bound $\|\psi\|_1$.  We think of $\psi$ as the
vertex potentials corresponding to electric flow with imbalance 
$\kron_s-\kron_t$.  By
an easy perturbation argument we can assume that no two vertex potentials
coincide. 

Index the vertices in $[n]$ by increasing potential as $\psi_1 < \cdots < \psi_n$.
Further, assume that $n$ is even and choose a median $c_0$ so that
$\psi_1<\cdots<\psi_{n/2} < c_0 < \psi_{n/2+1} < \cdots < \psi_n$.  
(If $n$ is odd, then
set $c_0$ to equal the pottential of the middle vertex.)

We aim to upper bound $\|\psi\|_1$, given as
$\|\psi\|_1 = \sum_v |\psi_v| 
  = \sum _{v:\psi_v > 0} \psi_v - \sum_{u:\psi_u < 0} \psi_u$.
Using that $\sum_w \psi_w = 0$, we get
$\|\psi\|_1 = 2\sum_{v:\psi_v>0} \psi_v = - 2\sum_{u:\psi_u < 0} \psi_u$.

Assume, without loss of generality, that $0 < c_0$, in which case
\begin{align}\label{def:ve:N}
\|\psi\|_1 = -2\sum_{u:\psi_u < 0} \psi_u
\le 2\sum_{i=1}^{n/2}|\psi_i - c_0| =: 2N
\end{align}
In what follows we aim to upper-bound $N$.

\textit{Flow cutting:}
Define a collection of cuts $(S_i,S_i^\cmpl)$ of the form 
$S_i = \{v : \psi_v \le c_i \}$, 
for integers $i\ge 0$, where $S_i$ will be the smaller side of the cut
by construction. Let $k_i$ be the number of edges cut by $(S_i,S_i^\cmpl)$ and
$p_{ij}$ be the length of the $j^{\mathrm{th}}$ edge across the same cut. The
cut points $c_i$, for $i\ge 1$, are defined according to
%
$c_i = c_{i-1}-\Delta_{i-1},\text{ where }
\Delta_{i-1}:=2\sum_j \dfrac{p_{i-1,j}}{k_{i-1}}$.
The last cut, $(S_{r+1},S_{r+1}^\cmpl)$, is the first cut in the sequence
$c_0,c_1,\dots,c_{r+1}$ with $k_{r+1}=0$ or, equivalently, $S_{r+1}=\emptyset$.

\textit{Bound on number of cuts:}
Let $n_i=|S_i|$.  The isoperimetric inequality for vertex expansion 
(\ref{ve:iso}) applied to
$(S_i,S_i^\cmpl)$ and the fact that $n_i\le n/2$, by construction, imply
\begin{align}\label{ineq:ve:iso}
\frac{k_i}{n_i} \ge \alpha.
\end{align}

Let $l_i$ be the number of edges crossing $(S_i,S_i^\cmpl)$ that do not
extend across $c_{i+1}$, i.e. edges that are not adjacent to $S_{i+1}$.
The choice $\Delta_i:=2\sum_j p_{ij} / k_i$ ensures that $l_i \ge k_i/2$.
These edges are supported on at least $l_i / d_{\max}$ vertices in $S_i$,
and therefore $n_{i+1} \le n_i - l_i/d_{\max}$. Thus,
\begin{align}\label{ineq:ve:shrink}
n_{i+1} 
  \le n_i - \frac{l_i}{d_{\max}} 
  \le n_i - \frac{k_i}{2d_{\max}} 
  \overset{(\ref{ineq:ve:iso})}{\le} n_i - \frac{\alpha n_i}{2d_{\max}}
= n_i\Big(1-\frac{\alpha}{2d_{\max}}\Big),
\end{align}
Combining inequality (\ref{ineq:ve:shrink}) 
with $n_0 = n / 2$, we get
\begin{align}\label{ineq:ve:n}
n_i \le \frac{n}{2}\Big(1-\frac{\alpha}{2d_{\max}}\Big)^i
\end{align}
The stopping condition implies $S_r\neq\emptyset$, or $n_r \ge 1$, and together
with (\ref{ineq:ve:n}) this results in 
\begin{align}\label{ineq:ve:r}
r \le \log_{1/\theta} \frac{n}{2},\text{ with }
\theta=1-\frac{\alpha}{2 d_{\max}}.
\end{align}

\textit{Amortization argument:}
Continuing from (\ref{def:ve:N}),
\begin{align}\label{ineq:ve:N}
N = \sum_{i=1}^{n/2}|\psi_i-c_0|
  \overset{(*)}{\le} \sum_{i=0}^r(n_i-n_{i+1})\sum_{j=0}^i \Delta_j,
\end{align}
where $(*)$ follows from the fact that
for every vertex $v\in S_i-S_{i+1}$ we can write
$|\psi_v-c_0|\le \sum_{j=0}^i\Delta_j$.

Because $B\Li(\kron_s-\kron_t)$ is a 
unit flow, we have the crucial (and easy to verify)
property that, for all $i$, $\sum_j p_{ij}=1$. In other words, the total
flow of ``simulatenous'' edges is 1. So,
\begin{align}
\Delta_i 
  = 2\sum_{j}\frac{p_{ij}}{k_i} 
  = \frac{2}{k_i}
  \overset{(\ref{ineq:ve:iso})}{\le} \frac{2}{\alpha n_i}
\end{align}
Now we can use this bound on $\Delta_j$ in (\ref{ineq:ve:N}),
\begin{align*}
\sum_{i=0}^r(n_i-n_{i+1})\sum_{j=0}^i \Delta_j
  \overset{(*)}{=} \sum_{i=1}^r n_i \Delta_i 
  \le \frac{2}{\alpha}\sum_{i=0}^r 1
  = \frac{2}{\alpha}(r+1),
\end{align*}
where to derive $(*)$ we use $n_{r+1}=0$. Combining the above inequality
with (\ref{ineq:ve:r}) concludes the proof.
\end{proof}

\section{Electric walk}
\label{sec:walk}

To every unit flow $f\in\matha{R}^E$, not necessarily electric, we associate a
random walk $W=W_0,W_1,\dots$ called the \textit{flow walk}, defined as follows.
Let $\sigma := B^*f$ and so $\sum_v \sigma_v = 0$. 
The walk starts at $W_0$, with 
\begin{align*}
\Pr\{W_0=v\}
  =\frac{2\cdot \max\{0,\sigma_v\}}{\sum_w |\sigma_w|}
  = \frac{\sum_w f_{v\to w} - \sum_w f_{w\to v}}
  {\sum_u\big(\sum_w f_{u\to w} - \sum_w f_{w\to u}\big)}
\end{align*}
If the walk is currently at $W_t$, the 
next vertex is chosen according to
$\Pr\big\{W_{t+1}=v\,|\,W_t=u\big\} 
  = \dfrac{f_{u\to v}}{\sum_{w}f_{u\to w}}$, where
%
\begin{align}
f_{u\to v} &= 
\begin{cases}
|f_{(u,v)}|,&\text{$(u,v)\in E$ and $f$ flows from $u$ to $v$} \\
0,&\text{otherwise.}
\end{cases} \label{dir:pot}
\end{align}
When the underlying flow $f$ is an electric flow, i.e. when $f=\el(B^*f)$, the
flow walk deserves the specialized name \textit{electric walk}.  We study two
aspects of electric walks here: (i)~stability against perturbations of the
vertex potentials, and (ii)~robustness against edge removal.

\subsection{Stability}

The set of vertex potential vectors $\pot^{[v]}=\Li \kron_v$,
for all $v\in V$, encodes all electric flows, as argued in~(\ref{fwd:flow}). 
In an algorithmic setting,
only approximations $\tilde{\pot}^{[v]}$ of these vectors are available. We ensure
that when these approximations are sufficiently good in an $\ell_2$ sense,
the path probabilities (and congestion properties) 
of electric walks are virtually unchanged. The
next theorem is proven in Appendix~\ref{sec:approx}:

\begin{theorem}\label{theo:approx}
Let $\tilde{\pot}^{[v]}$ be
an approximation of $\pot^{[v]}$, for all $v\in G$, in the sense that
\begin{align}
\|\pot^{[v]}-\tilde{\pot}^{[v]}\|_2 \le \nu,
\text{ for all $v\in V$, with } \label{ineq:pot:approx}
\nu = n^{-A},
\end{align}
where $A>4$ is a constant. Then for every electric walk,
defined by vertex potentials $\pot=\sum_v \alpha_v \pot^{[v]}$,
the corresponding ``approximate'' walk, defined by vertex potentials
$\tilde{\pot}=\sum_v \alpha_v\tilde{\pot}^{[v]}$, induces a 
distribution over paths $\gamma$ with
\begin{align}
\sum_\gamma \Big|\Pr_\pot\{W=\gamma\} 
  - \Pr_{\tilde{\pot}}\{W=\gamma\}\Big| \le O(n^{2-\frac{A}{2}}), 
  \label{stat:diff}
\end{align}
where $\gamma$ ranges over all paths in $G$, and
$\Pr_\pot\{W=\gamma\}$ denotes the probability of $\gamma$ under $\pot$ 
(respectively for $\Pr_{\tilde{\pot}}\{W=\gamma\}$).
\end{theorem}

As shown in Theorem~\ref{theo:power}, 
the Power Method affords us any sufficiently large
exponent $A$, say $A=5$, without sacrificing efficiency in terms
of distributed computation time. In this case,
the following corollary asserts that routing with approximate potentials
preserves both the congestion properties of the exact electrical flow as well
as the probability of reaching the sink.

\begin{corollary}\label{coro:approx}
Under the assumptions of Theorem~\ref{theo:approx} and $A=5$, 
the electric walk
defined by vertex potentials 
$\tilde{\pot}=\tilde{\pot}^{[s]}-\tilde{\pot}^{[t]}$ reaches $t$ 
with probability $1-o_n(1)$. Furthermore, 
for every edge $(u,v)$ with
non-negligible load, i.e. $|\tilde{\pot}_u-\tilde{\pot}_v|=\omega(n^{-2})$, 
we have
$|\tilde{\pot}_u-\tilde{\pot}_v|\to_n |\pot_u-\pot_v|$, 
where $\pot=\pot^{[s]}-\pot^{[t]}$.
\end{corollary}

\subsection{Robustness}
We prove Theorem~\ref{theo:robust1} here, since its proof interestingly 
relies on the flow cutting techniques developed in this paper. 
Theorem~\ref{theo:robust2} is proved in Appendix~\ref{sec:robust-proof}.

\begin{proof}[Proof of Theorem~\ref{theo:robust1}]
For the first part, let $\{(u_1,v_1),\dots,(u_k,v_k)\}=Q_p$ and let 
$p_i =
|f^{[s,t]}_{(u_i,v_i)}|=|(\kron_{u_i}-\kron_{v_i})^*\Li(\kron_s-\kron_t)|$.
Consider the embedding $\zeta:V\to\matha{R}$, defined by
$\zeta(v)=\kron_v^*\Li(\delta_s-\delta_t)$. 
Assume for convenience that $\zeta(u_i)\le \zeta(v_i)$ for all $i$.
Let $\zeta_{\min}=\min_v\,\zeta(v)$ and 
$\zeta_{\max}=\max_v\,\zeta(v)$. 

Choose $c$ uniformly at random
from $[\zeta_{\min},\zeta_{\max}]$ and let 
$X_i=p_i\cdot\I\{\zeta(u_i)\le c\le \zeta(v_i)\}$, where $\I\{\cdot\}$
is the indicator function. Observe that the random variable
$X=\sum_i X_i$ equals the total electric flow of all edges in $Q_p$ cut by $c$.
Since these edges are concurrent (in the electric flow) by construction,
we have $X\le 1$. On the other hand,
\begin{align*}
\E X = \sum_i p_i \cdot \Pr\Big\{\zeta(u_i)\le c\le \zeta(v_i)\Big\} 
  \ge \sum_i p_i \frac{p_i}{\zeta_{\max}-\zeta_{\min}} 
  \ge \sum_i \frac{\lambda p_i^2}{2}
  \ge k \frac{\lambda p^2}{2}
\end{align*}
Combining this with $\E X \le 1$ produces $|Q_p|\le 2/(\lambda p^2)$.

For the second part, 
$kp \le \sum_{e\in Q_p} |f^{[s,t]}_e| \le \sum_{e\in E} |f^{[s,t]}_e|
= \|B\Li(\kron_s-\kron_t)\|_1\le2d_{\max}\cdot \|\Li\|_{1\to 1}$.
This gives $|Q_p|\le 2d_{\max}\cdot\|\Li\|_{1\to 1}/ p$.
\end{proof}

\section{Conclusions}
\label{sec:concl}

Our main result in Theorem~\ref{theo:main} attests to the good congestion
properties on graphs of bounded degree and high vertex expansion, i.e.
$\alpha=O(1)$. A variation on the proof of this theorem establishes a similar
bound on $\eta_\el$, however, independent of the degree bound and 
as a function of the \textit{edge expansion} 
$\beta=\min_{S\subseteq V}
\frac{\vvol(E(S,S^\complement))}{\min\{\vvol(S),\vvol(S^\complement)\}}$, where
$\vvol(S):=\sum_{v\in S}\sum_{u:u\sim v} w_{u,v}$ and 
$\vvol(E(S,S^\complement)):=\sum_{(u,v):u\in S,v\in S^\complement} w_{u,v}$.

The bounded degree assumption is also implicit in our computational procedure
in that all vertices must know an upper bound on $d_{\max}$ in order to 
apply $M$ in Theorem~\ref{coro:power}. Using a generous bound, 
anything $\omega(1)$, on $d_{\max}$ is
bad because it slows down the mixing of the power polynomial. To avoid this
complication, one must use a symmetrization trick outlined in
Appendix~\ref{sec:symmetrize}.

We conclude with a couple of open questions.
A central concern, widely-studied in social-networks, are Sybil
Attacks~\cite{sybil}. These can be modeled as graph-theoretic noise, as defined
in~\cite{noise}. It is interesting to understand how such noise affects
electric routing.
We suspect that any $O(\ln n)$-competitive oblivious routing scheme,
which outputs its routes in the ``next hop'' model, must maintain 
$\Omega(n)$-size routing tables at every vertex. In the \textit{next hop}
model, every vertex $v$ must be able to answer the question ``What is the
flow of the $(s,t)$-route in the neighborhood of $v$?'' in time 
$O(\polylog(n))$, 
using its own
routing table alone and for every source-sink pair $(s,t)$.

\bibliographystyle{plain}
\bibliography{main}


\appendix
\section{Proof of Theorem~\ref{theo:cong:univ}}
\label{sec:cong:u}

\begin{proof}[Proof of Theorem~\ref{theo:cong:univ}]
The upper bound follows from:
\begin{align}
\eta
    \overset{(\ref{ineq:eta:ub})}{\le} \|\Pi\|_{1\rightarrow 1}
    \overset{(\ref{ineq:pi})}{\le} m^{1/2}\cdot\|\Pi\|_{2\rightarrow 2}
    \overset{(\ref{ineq:pi:proj})}{=} m^{1/2} \label{ineq:univ}
\end{align}

The second step is justified as follows
\begin{align}
\|\Pi\|_{1\rightarrow 1}
    = \max_{e} \|\Pi\kron_e \|_1
    \le m^{1/2} \cdot \max_e \|\Pi \kron_e\|_2
    \le m^{1/2} \cdot \|\Pi\|_{2\rightarrow 2}. \label{ineq:pi}
\end{align}
The third step is the assertion
\begin{align}
\|\Pi\|_{2\to 2} \le 1 \label{ineq:pi:proj},
\end{align}
which follows from the (easy) fact that $\Pi$ is a projection,
shown by Spielman, et al. in Lemma~\ref{lem:spiel}.

The lower bound is achieved by a graph obtained by
gluing the endpoints of $\sqrt{n}$ copies of a path
of length $\sqrt{n}$ and a single edge. Routing
a flow of value $\sqrt{n}$ between these endpoints incurs
congestion $\sqrt{n}/2$.
\end{proof}

\begin{lemma}[Proven in~\cite{spiel}]\label{lem:spiel}
$\Pi$ is a projection; $\vim(\Pi)=\vim(W^{1/2}B)$;
The eigenvalues of $\Pi$ are 1 with multiplicity $n-1$
and $0$ with multiplicity $m-n+1$; and $\Pi_{e,e}=\|\Pi\kron_e\|^2$.
\end{lemma}

\section{Proof of Theorem~\ref{theo:m:lb}}
\label{sec:norm:lower}

\begin{proof}[Proof of Theorem~\ref{theo:m:lb}] 
Let $s\neq t$ be a pair of vertices in $G$
at distance $D$. We consider the
flow $f=B\Li(\kron_s-\kron_t)$. Set $\psi=\Li(\kron_s-\kron_t)$, and note that
we can use $\|f\|_1$ as a lower bound on $\|\Li\|_{1\to1}$,
\begin{align*}
\|f\|_1 
  = \sum_{(u,v)}|\psi_u - \psi_v|
  \le d_{\max} \sum_{v} |\psi_v| 
  \le d_{\max} \|\Li (\kron_s-\kron_t)\|_1
  \le \frac{d_{\max}}{2}\|\Li\|_{1\to 1}.
\end{align*}
Now, let $\{\pi_i\}_i$ be a path decomposition of $f$ and let
$l(\pi_i)$ and $f(\pi_i)$ denote the length and value, respectively, of
$\pi_i$. Then,
\begin{align*}
\|f\|_1 
  = \sum_{(u,v)}|\psi_u-\psi_v| 
  = \sum_i l(\pi_i) f(\pi_i)
  \ge D \sum_i f(\pi_i)
  = D.&\qedhere
\end{align*}
\end{proof}

\section{Proof of Theorem~\ref{theo:approx}}
\label{sec:approx}

\begin{proof}[Proof of Theorem~\ref{theo:approx}]
\textit{Notation:}
Note that in this proof we use the notation of~(\ref{dir:pot}).
So, for a potential vector $\psi$, we have 
$\psi_{u\to v} := \psi_u-\psi_v$ if $(u,v)$ is an edge and $\psi_u \ge \psi_v$,
and $\psi_{u\to v}:=0$ otherwise. So, for example, the potential 
difference on $(u,v)$
can be written as $\psi_{u\to v}+\psi_{v\to u}$. On the other hand,
we use the single letter edge notation $\pot_e$ to denote the signed
(according to $B$) potential difference on $e$, 
so $\pot_e:=\pot_u-\pot_v$ if $B_e=\delta_u-\delta_v$. Let $D$ be the 
maximum degree.

\textit{Edge approximation:}
Fix any unit electric flow, defined by potentials 
$\pot:=\sum_v \alpha_v \pot^{[v]}$, and write its approximation
as $\tilde{\pot}:=\sum_v \alpha_v \tilde{\pot}^{[s]}$.
All unit flows can be so expressed under the restriction
that $\sum_v \alpha_v = 0$ and $\sum_v |\alpha_v|=2$.
The approximation condition~(\ref{ineq:pot:approx}) 
combined with Lemma~\ref{lem:l2:apx} gives us, for every edge $e=(u,v)$,
\begin{align*}
|\pot_e - \tilde{\pot}_e| 
  &= \big|\sum_v\alpha_v(\pot^{[v]}_e - \tilde{\pot}^{[v]}_e)\big| \\
  &\le \sum_v |\alpha_v|\cdot\big|\pot^{[v]}_e - \tilde{\pot}^{[v]}_e\big| \\
  &\le \sum_v |\alpha_v| \cdot 2\nu 
    &\text{apply Lemma~\ref{lem:l2:apx}}\\
  &= 4\nu
\end{align*}
We call this the \textit{additive edge approximation condition}
\begin{align}
\pot_e -4\nu \le \tilde{\pot}_e \le \pot_e + 4\nu \label{ineq:edge:add}
\end{align}

Now, consider a fixed path $\gamma$ 
along the electric flow defined by $\pot$,
traversing vertices $w_0,w_1,\dots,w_k$. Let $\Pr_\pot\{W=\gamma\}$
and $\Pr_{\tilde{\pot}}\{W=\gamma\}$ denote the probability of this path
under the potentials $\pot$ and $\tilde{\pot}$, respectively. 
In most of what follows, we build machinery to relate one to the other.

\textit{Path probabilities:}
For a general unit flow (not necessarily an $(s,t)$-flow), defined by 
vertex potentials $\psi$, $\Pr_{\psi}\{W=\gamma\}$ equals
\begin{align} \label{prob:path}
  \Pr_\psi\{W_0=w_0\}\Bigg(\prod_{i=0}^{k-1} 
    \Pr_\psi\{W_{i+1}=w_{i+1}\,|\,W_i=w_i\}\Bigg)
    \Pr_\psi\{W_\infty=w_k\,|\,w_k\},
\end{align}
where next we explain each factor in turn.

The first, $\Pr_{\psi}\{W_0=w_0\}$, is the probability that 
the walk starts from $w_0$, and is expressed as
\begin{align}
\Pr_\psi\{W_0=w_0\} 
  = \max\Big(0,
  \sum_u \psi_{w_0\to u} - \sum_u \psi_{u\to w_0} \Big). \label{prob:start}
\end{align}

The second and trickiest, $\Pr_\psi\{W_{i+1}=w_{i+1}\,|\,W_i=w_i\}$, is
the probability that having reached $w_i$ the walk traverses the edge
leading to $w_{i+1}$, and
$\sum_u \psi_{w_i\to u} \ge \sum_u \psi_{u\to w_i}$, we write
\begin{align}
\Pr_\psi\{W_{i+1}=w_{i+1}\,|\,W_i=w_i\}
  = \frac{\psi_{{w_i\to w_{i+1}}}}
  {\displaystyle\max\Big(\sum_u \psi_{u\to w_i},\sum_u \psi_{w_i\to u}\Big)}.
  \label{prob:trans}
\end{align}
To grasp the meaning of the denominator, note that the quantity
$|\sum_u \psi_{u\to w_i} - \sum_u \psi_{w_i\to u}|$ is the magnitude
of the in or out flow (depending on the case) at $w_i$.

The third, $\Pr_\psi\{W_\infty=w_k\,|\,w_k\}$, is the probability that the walk
ends (or exits) at $w_k$ conditioned on having reached $w_k$, and
\begin{align}
\Pr_\psi\{W_\infty=w_k\,|\,w_k\} 
  = \max \Big(0,\sum_u \psi_{u\to w_k} - \sum_u \psi_{w_k\to u}\Big). 
  \label{prob:end}
\end{align}

Next, we are going to find multiplicative bounds for all three factors
by focusing on ``dominant'' paths, and discarding ones with overall negligible
probability.

\textit{Dominant paths:}
It is straightforward to verify (from first principles) 
that the probability that an edge $(u,v)$
occurs in the electric walk equals $|\pot_e| = \pot_{u\to v}+\pot_{v\to u}$.
We call an edge \textit{short} if 
$|\pot_e|\le \epsilon$, where 
the exact asymptotic of $\epsilon > 0$ is determined later,
but for the moment $\nu \ll e \ll 1$. We restrict our
attention to \textit{dominant} paths $\gamma$ 
that traverse no short edges, and have
$\Pr_\pot\{W_0=w_0\}\ge\epsilon$ and 
$\Pr_\pot\{W_\infty=w_k\,|\,w_k\}\ge\epsilon$.

Indeed, by a union bound, the probability that
the electric walk traverses a non-dominant path 
is at most $2n\epsilon+n^2\epsilon$. This
will be negligible and such paths will be of no interest. In summary,
\begin{align}
\Pr_{\pot}\{\text{$W$ dominant}\} \ge 1 - 2n\epsilon - n^2\epsilon
  \label{prob:dom}
\end{align}
We now
condition on the event that $\gamma$ is dominant. 

The no short edge condition gives
$\epsilon\le|\pot_e|\le 1$, and using~(\ref{ineq:edge:add})
we derive the stronger \textit{multiplicative edge approximation condition}
\begin{align}
\frac{1}{\sigma} \le \frac{\tilde{\pot}_e}{\pot_e} \le \sigma,
\text{ where } \sigma = 1 + \frac{8D\nu}{\epsilon}, \label{ineq:edge:mul}
\end{align}
which holds as long as $\epsilon\ge 4\nu$, as guaranteed 
by the asymptotics of $\epsilon$. Also note that the latter condition ensures
that $\pot_e$ and $\tilde{\pot}_e$ have the same sign.
An extra factor of $2D$ is included in $\sigma$ with
foresight.

For the first factor~(\ref{prob:start}), we have
\begin{align} \label{bound:start}
\Pr_{\tilde{\pot}}\{W_0=w_0\} 
  &= \sum_u \tilde{\pot}_{w_0\to u} - \sum_u \tilde{\pot}_{u\to w_0} \\
  &\ge \sum_u \pot_{w_0\to u} - \sum_u \pot_{u\to w_0} - 4D\nu 
    &\text{use~(\ref{ineq:edge:add})} \notag \\
  &\ge \Pr_{\pot}\{W_0=w_0\} \Big(1-\frac{4D\nu}{\epsilon}\Big)
    &\text{use $\Pr_{\pot}\{W_0=w_0\}\ge \epsilon$} \notag \\
  &\ge \frac{1}{\sigma}\Pr_{\pot}\{W_0=w_0\}
    &\text{use $\epsilon\le 1/2$}. \notag
\end{align}

For the second factor~(\ref{prob:trans}), assume 
$\sum_u \psi_{u\to w_i} \ge \sum_u \psi_{w_i\to u}$. An identical
argument holds in the other case. 
Abbreviate
\begin{align*}
\Pr_{\tilde{\pot}}\{w_{i+1}\,|\,w_i\}
  := \Pr_{\tilde{\pot}}\{W_{i+1}=w_{i+1}\,|\,W_i=w_i\}.
\end{align*}
Path dominance implies
$\sum_u \pot_{u\to w_i} \ge \epsilon$, and so
\begin{align} \label{bound:trans}
\Pr_{\tilde{\pot}}\{w_{i+1}\,|\,w_i\}
  &= \frac{\tilde{\pot}_{w_i\to w_{i+1}}}
    {\displaystyle\sum_u \tilde{\pot}_{u\to w_i}} \\
  &\ge \frac{\sigma^{-1}\cdot\pot_{w_i\to w_{i+1}}}
    {\displaystyle\sum_u \pot_{u\to w_i}+4D\nu}
    &\text{use~(\ref{ineq:edge:mul}) and~(\ref{ineq:edge:add})} \notag \\
  &\ge \sigma^{-2} \frac{\pot_{w_i\to w_{i+1}}}{\sum_u \pot_{u\to w_i}}
    &\text{use $\epsilon \le 1/2$ and 
      $\sum_u \pot_{u\to w_i} \ge \epsilon$} \notag \\
  &= \frac{1}{\sigma^2} \Pr_{\pot}\{w_{i+1}\,|\,w_i\}. \notag
\end{align}

For the third factor~(\ref{prob:end}), similarly to the first, we have
\begin{align} \label{bound:end}
&\Pr_{\tilde{\pot}}\{W_\infty=w_k\,|\,w_k\}
  = \\
  &\qquad= \sum_u \tilde{\pot}_{u\to w_k} 
    - \sum_u \tilde{\pot}_{w_k\to u} \notag \\
  &\qquad\ge \sum_u \pot_{u\to w_k} - \sum_u \pot_{w_k\to u} -4D\nu
    &\text{use~(\ref{ineq:edge:add})} \notag \\
  &\qquad\ge \Pr_{\pot}\{W_\infty=w_k\,|\,w_k\}\Big(1-\frac{4D\nu}{\epsilon}\Big)
    &\text{use $\Pr_{\pot}\{W_\infty=w_k\,|\,w_k\}\ge \epsilon$} \notag \\
  &\qquad\ge \frac{1}{\sigma} \Pr_{\pot}\{W_\infty=w_k\,|\,w_k\}
    &\text{use $\epsilon\le 1/2$}. \notag
\end{align}

\textit{Dominant path bound:}
We now obtain a relation between $\Pr_\pot\{W=\gamma\}$ and 
$\Pr_{\tilde{\pot}}\{W=\gamma\}$ by combinging the bounds
(\ref{bound:start}), (\ref{bound:trans}) and (\ref{bound:end}) 
with (\ref{prob:path}):
\begin{align}
\frac{\Pr_{\tilde{\pot}}\{W=\gamma\}}{\Pr_\pot\{W=\gamma\}} 
  &\ge \frac{1}{\sigma^{2n+2}} 
    &\text{apply bounds, and path length $\le n$}
      \label{approx:theta} \\
  &\ge \Big(1-\frac{8D\nu}{\epsilon}\Big)^{2n+2}
    &\text{use $\sigma^{-1} \ge 1 - 8D\nu/\epsilon$} \notag \\
  &\ge \exp\Big(-\frac{16 D\nu}{\epsilon}\Big)^{2n+2}
    &\text{use $1-x\ge e^{-2x}$} \notag \\
  &=: \theta \notag
\end{align}

\textit{Statistical difference:}
Abbreviate $p(\gamma):=\Pr_\pot\{W=\gamma\}$ and
$q(\gamma):=\Pr_\pot\{W=\gamma\}$. Below, $\gamma$ iterates
through all paths, $\zeta$ iterates through dominant 
paths and $\xi$ iterates through non-dominant paths. 
We bound the statistical difference~(\ref{stat:diff}),
using~(\ref{approx:theta}) which says $q(\zeta)\ge\theta\cdot p(\zeta)$,
\begin{align}
&\sum_{\gamma}|p(\gamma)-q(\gamma)| = \notag\\
  &\qquad= \sum_\zeta|p(\zeta)-q(\zeta)|+\sum_\xi|p(\xi)-q(\xi)| \notag \\
  &\qquad\le \sum_\zeta|(1-\theta)p(\zeta)-\big(q(\zeta)-\theta p(\zeta)\big)|
    +\sum_\xi q(\xi) 
    &\text{use $\theta<1$} \notag \\
  &\qquad\le \sum_\zeta|q(\zeta)-\theta p(\zeta)|+\sum_\xi q(\xi) \notag \\
  &\qquad= 1 - \sum_\zeta p(\zeta) \notag \\
  &\qquad= (1-\theta) + \theta\sum_\zeta p(\zeta) \label{stat:diff:2}
\end{align}
In this final step, we pin-point the asymptotics of $\epsilon$ that
simultaneously minimize the two terms of~(\ref{stat:diff:2}). In the following,
we parameterize $\epsilon=n^{-B}$ and use~(\ref{prob:dom}),
\begin{align*}
&(1-\theta) + \theta\sum_\zeta p(\zeta) = \\
  &\qquad= 1-\exp\Big(-\frac{16D\nu}{\epsilon}\Big)^{2n+2}
    + \exp\Big(-\frac{16D\nu}{\epsilon}\Big)^{2n+2}
      \big(2n\epsilon+n^2\epsilon\big) \\
  &\qquad= 1-\exp O\big(-Dn^{B-A+1}\big) 
    + n^{2-B} \cdot\exp O\big(-Dn^{B-A+1}\big) \\
  &\qquad= O\big(Dn^{B-A+1}\big)
    + n^{2-B} \cdot\exp O\big(-Dn^{B-A+1}\big), 
    \qquad\text{use $1-e^{-x}\le x$}\\
  &\qquad= O\big(n^{B-A+2}\big) + O\big(n^{2-B}\big) 
    \qquad\text{use $D\le n$}\\
  &\qquad= O\big(n^{2-\frac{A}{2}}\big),
    \qquad\text{set $B=A/2$.} &\qedhere
\end{align*}
\end{proof}

\begin{lemma}\label{lem:l2:apx}
If $x,y\in\ell_2$ and $\|x-y\|_2\le\nu$, then for all $i\neq j$,
\begin{align*}
(x_i-x_j)-2\nu \le y_i - y_j \le (x_i-x_j)+2\nu.
\end{align*}
\end{lemma}

\begin{proof}
We have $(x_i-y_i)^2 \le \|x-y\|_2^2\le \nu^2$, implying $|x_i-y_i|\le \nu$.
Similarly for $j$. Combining the two proves the lemma.
\end{proof}

\section{Proof of Theorem~\ref{theo:stretch}}
\label{sec:latency-proof}

\begin{proof}[Proof of Theorem~\ref{theo:stretch}]
Let $X^{[s,t]}_e$ be the indicator that edge $e$ participates in the
electric walk between $s$ and $t$. Then the latency can be expressed as
\begin{align*}
\max_{s\neq t} \sum_e \E X^{[s,t]}_e
  &= \max_{s\neq t} 
    \sum_{(u,v)}\big|(\delta_u-\delta_v)\Li(\delta_s-\delta_t)\big| \\
  &= \max_{s\neq t} \|B\Li(\delta_s-\delta_t)\|_1 
  = \|B\Li B^*\|_{1\to 1} = \|\Pi\|_{1\to 1}.
\end{align*}
The latter is bounded by Theorem~\ref{theo:ve:ub} and 
$\|\Pi\|_{1\to 1} \le m^{1/2}$, as in~(\ref{ineq:univ}) e.g.
\end{proof}
\begin{remark}
For expanders, this theorem is not trivial. In fact, there exist
path realizations of the electric walk which can traverse up to $O(n)$ edges.
Theorem~\ref{theo:stretch} asserts that this happens with small probability. On
the other hand, in a bounded-degree expander, even if $s$ and $t$ are adjacent
the walk will still take a $O(\log n)$-length path with constant probability.
\end{remark}

\section{Proof of Theorem~\ref{theo:robust2}}
\label{sec:robust-proof}

\begin{proof}[Proof of Theorem~\ref{theo:robust2}]
Let $f_{\opt}$ be a max-flow routing of the uniform demands and 
let $\theta$ be the fraction of the demand set that is routed by $f_{\opt}$.
The Multi-commodity Min-cut Max-flow Gap Theorem (Theorem 2, in~\cite{mcmf})
asserts
\begin{align*}
O(\ln n)\cdot \theta 
  \ge \min_{S\subset V} \frac{|E(S,S^\cmpl)|}{|S|\cdot|S^\cmpl|}
  \ge \frac{1}{n}\cdot 
    \min_{S\subset V} \frac{|E(S,S^\cmpl)|}{\min\{|S|,|S^\cmpl|\}}
  = \frac{\alpha}{n}
\end{align*}
Thus the total demand flown by $f_{\opt}$ is no less than
$\theta \binom{n}{2} \ge \Omega(\alpha n / \ln n)$. Normalize $f$ (by scaling)
so it routes the same demands as $f_{\opt}$. If $k$ edges are removed,
then at most $\eta k$ flow is removed from $f$, which is at most a fraction
$\eta k \cdot O(\ln n / \alpha n)$ of the total flow.
Substitute $x=k/m$ and use $m\le d_{\max}n$ to complete the proof.
\end{proof}

\section{Proof of Theorem~\ref{coro:power}}

Theorem~\ref{coro:power} is implied by the following theorem
by specializing $\epsilon=O(n^{-5})$:

\begin{theorem}\label{theo:power}
Let $G$ be a graph, whose Laplacian $L$ has smallest eigenvalue $\lambda$
and whose maximum degree is $D$. Then, for every $y$ with $\|y\|_2=1$
the vector $x=\Li y$ can be approximated using
\begin{align*}
\tilde{x} = \frac{1}{2D}\sum_{i=0}^d \big(I-\frac{L}{2D}\big)^i y,
\end{align*}
so that for every $\epsilon>0$, 
\begin{align*}
\|x-\tilde{x}\|_2 \le \epsilon,
\text{ as long as } 
d \ge \Omega(1)\cdot
 \ln \frac{1}{\lambda\epsilon D}\cdot
  \Big(\ln\frac{1}{1-\lambda}\Big)^{-1}.
\end{align*}
\end{theorem}

\newcommand{\Ni}{N^{\dag}}
\begin{proof}[Proof of Theorem~\ref{theo:power}] 
We normalize $L$ via $N=L/\tau$ (and so $L^{-1}=N^{-1}/\tau$),
where $\tau=2D$. Since $\tau= 2D\ge \lambda_{\max}(L)$, the
eigenvalues of $N$ are in $[0,1]$. 
In this case, the Moore-Penrose inverse of $N$ is
given by $\Ni=\sum_{i=0}^{\infty} (I-N)^i$.  
Set $\Ni_0=\sum_{i=0}^d (I-N)^i$ and $\Ni_1 = \Ni - \Ni_0$.
Our aim is to minimize $d$ so that 
\begin{align*}
\|x-\tilde{x}\|_2 
  = \Big\|\frac{\Ni_0+\Ni_1}{\tau}y - \frac{\Ni_0}{\tau}y\Big\|_2
  = \Big\|\frac{\Ni_1}{\tau}y\Big\|_2 
  \le \Big\|\frac{\Ni_1}{\tau}\Big\|_{2\to 2}
  \le \epsilon,
\end{align*}
where $\|A\|_{2\to 2}:=\sup_{x\neq 0}\|Ax\|_2/\|x\|_2$ 
denotes the matrix spectral norm. Set
$\kappa:=\tau/\lambda_{\min}$, so that $\kappa^{-1}$ is the smallest
eigenvalue of $N$,
\begin{align}
\|\Ni_1\|_{2\to 2} = \big\|\sum_{i=d+1}^{\infty} (I-N)^i\big\|_{2\to 2}
  &\le \sum_{i=d+1}^\infty\|(I-N)^i\|_{2\to 2} \notag \\
  &\le \sum_{i=d+1}^\infty\left(1-\kappa^{-1}\right)^i
  = (1-\kappa^{-1})^{d+1} \kappa \label{pow:deg}
\end{align}
Setting~(\ref{pow:deg}) less than $\tau\epsilon$ gives
\begin{align*} 
d \ge \frac{\ln \kappa/(\tau\epsilon)}{\ln \kappa/(\kappa-1)}. &\qedhere
\end{align*}
\end{proof}

\section{Symmetrized algorithm}
\label{sec:symmetrize}

In this section we discuss how to modify the computational procedure, given in
the Section~\ref{sec:intro}, in order to apply it to graphs of 
unbounded degree.
The described algorithm for computing $\pot^{[w]}=\Li\kron_w$ relies on the 
approximation of $\Li$ via the Taylor series 
$\frac{1}{1-x}=\sum_{i=0}^\infty (1-x)^i$. 
The series converges only when $\|x\|_2<1$,
which is ensured by setting $x=\frac{L}{2d_{\max}}$, 
and using that $\|L\|_{2\to 2} < 2d_{\max}$.
Thus we arrive at 
$2d_{\max}\cdot \Li = \sum_{i=0}^\infty (I-\frac{L}{2d_{\max}})^i$.
This approach continues to work if we replace $d_{\max}$ with any upper bound
$h_{\max}\ge d_{\max}$, obtaining $\Li = \frac{1}{2h_{\max}} \sum_{i=0}^\infty
(I - M)^i$ where $M=\frac{L}{2h_{\max}}$, however this is done at the expense
of slower convergence of the series. Since in a distributed setting all
vertices must agree on what $M$ is, a worst-case upper bound $h_{\max}=n$ must
be used, which results in a prohibitively slow convergence even for expander
graphs.

Instead, we pursue a differnt route. Let $\mathcal{L}= D^{-1/2}LD^{-1/2}$
be the \textit{normalized Laplacian} of $G$, where $D\in\matha{R}^{n\times n}$
is diagonal with $D_{v,v}=\deg(v)$. One always has $\|\mathcal{L}\|_{2\to 2}
\le 2$ (Lemma 1.7 in~\cite{chung}) while at the same time
$\lambda_{\min}(\mathcal{L})
\ge \max \big\{\frac{\beta^2}{2},\frac{\alpha^2}{4d_{\max} +
2d_{\max}\alpha}\big\}$ (Theorems 2.2 and 2.6 in~\cite{chung}), 
where $\alpha$ and $\beta$ are the vertex- and
edge-expansion of $G$, respectively.
Set $M=\mathcal{L}/3$, so that $\|M\|_{2\to 2} < 1$.
Recall that the aim of our distributed procedure is to 
compute $\pot^{[w]}_u$ at $u$ (for all $w$). We achieve this using
the following:
\begin{align*}
\pot^{[w]}_u 
  = \kron_u^* \Li \kron_w
  = \kron_u^* D^{-1/2}\frac{M^\dag}{3}D^{-1/2} \kron_w
  = \frac{\kron_u^*}{\sqrt{\deg(u)}} 
    \frac{\sum_{i=0}^\infty (I-M)^i}{3} 
    \frac{\kron_w}{\sqrt{\deg(w)}} 
\end{align*}
The key facts about the series in the left-hand side are that
(i) it converges quickly when $G$ is an expander and (ii) all vertices
can compute $M$ locally, in particular, without requiring any global knowledge
like e.g. an upper bound on $d_{\max}$.

\end{document}